\documentclass[11pt]{article}

\usepackage{amssymb}
\usepackage{amsmath}
\usepackage{amsfonts}
\usepackage{eucal}
\usepackage{graphicx}

\newtheorem{theorem}{Theorem}[section]

\newtheorem{rem}[theorem]{Remark}

\newenvironment{proof}[1][Proof]{\begin{trivlist}
\item[\hskip \labelsep {\bfseries #1}]}{\end{trivlist}}

\begin{document}

\title{A Model of Market Limit Orders By Stochastic PDE's,\\
    Parameter Estimation, and Investment Optimization}
\author{Zhi Zheng, Richard B. Sowers}
\date{July 1, 2012}

\maketitle

\begin{abstract}

In this paper we introduce a completely continuous and time-variate model of the evolution of market limit orders based on the existence, uniqueness, and regularity of the solutions to a type of stochastic partial differential equations obtained in~\cite{zb}. In contrary to several models proposed and researched in literature, this model provides complete continuity in both time and price inherited from the stochastic PDE, and thus is particularly suitable for the cases where transactions happen in an extremely fast pace, such as those delivered by high frequency traders (HFT's). 

We first elaborate the precise definition of the model with its associated parameters, and show its existence and uniqueness from the related mathematical results given a fixed set of parameters. Then we statistically derive parameter estimation schemes of the model using maximum likelihood and least mean-square-errors estimation methods under certain criteria such as AIC to accommodate to variant number of parameters . Finally as a typical economics and finance use case of the model we settle the investment optimization problem in both static and dynamic sense by analysing the stochastic (It\^{o}) evolution of the utility function of an investor or trader who takes the model and its parameters as exogenous. Two theorems are proved which provide criteria for determining the best (limit) price and time point to make the transaction.
\end{abstract}

\section{Introduction}

In literature there have been researches on the modeling of market limit orders and their execution such as~\cite{mloem} and~\cite{mlodm}, most of which are based on discrete settings, for instance, models based on Poisson processes and/or queuing theory. However, because of rapid technological evolution which brings about ultra-fast microprocessors and hardware, trading behaviors and patterns involved with a large amount of limit order creation, transaction, and cancellation within a short period of time, such as high frequency trading (HFT), have become quite popular and tend to have a heavy impact on the mechanisms of price discovery and formulation. Consequently, a continuous and dynamic model of the evolution of limit orders in a particular market and their dynamics is proposed, which is described by a type of stochastic partial differential equations, Stefan equations, with a number of model parameters to be estimated given a real dataset.

\begin{figure}
\centering
\includegraphics[width=100mm]{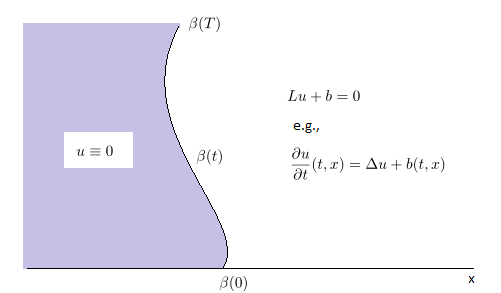}
\caption{An Illustration of Moving Boundary PDE's}
\label{Fig:mbp1}
\end{figure}

Such type of equations of $u(t,x)$ describes the behavior of a system that consists of two phases, as illustrated in Figure~\ref{Fig:mbp1}, where $\beta(t)$ is a moving boundary which is part of the solution and must be solved simultaneously with $u(t,x)$. As can be seen in Figure~\ref{Fig:mbp1}, in the region to the left of the moving boundary (namely, the set $\{(t,x)\in[0,\infty)\times\mathbb R:x\le\beta(t)\}$) $u$ is constantly set to $0$; on the right side of the boundary (the set $\{(t,x)\in[0,\infty)\times\mathbb R:x>\beta(t)\}$) $u$ is described by a PDE of the general form $Lu+b=0$ where $L$ is a predefined second-order differential operator.

For a moving boundary PDE problem, in addition to the regular boundary condition such as the Dirichlet condition  $u(t,\beta(t))=0$, there is always an extra boundary condition that describes the dynamics at the moving boundary, for instance, the \emph{Stefan boundary condition}
 \begin{equation}\label{Eq:SCondition}
  \frac{\partial u}{\partial x}(t,\beta(t)+)=\rho\dot\beta(t).
  \end{equation}

The type of moving boundary PDE's we shall use in modeling of market limit order evolution is the \emph{Stefan problems}, where $L$ is a heat or parabolic operator (for instance, $L:=-\partial/\partial t+\partial^2/\partial x^2$) with the Stefan boundary condition. Such type of problems has a variety of applications. For instance, in physics, they model the phenomena such as ice melting with the Stefan condition describing the heat balance at the interface (the moving boundary, see~\cite{fri64}); in finance they model the valuation of American options with the PDE derived from the Black-Scholes formula and the moving boundary describing the early exercise price boundary (see Lemma 7.8, Chapter 2,~\cite{mmf}). As it turns out in this paper such equations are also perfectly suitable for modeling the dynamics of high frequency limit order transactions and their evolution, providing statistical parameter estimation methods, and formulating and solving the investment optimization problem based on a typical utility function. Note that~\cite{zb} provides the mathematical foundation and develops essential techniques that enable us to justify the existence and uniqueness of such type of equations, while obtaining its essential regularities.

In Section~\ref{S:MLO:Modeling} we model the evolution of limit orders in a particular market by Stefan equations according to the following facts, and show the existence of such model based on the mathematical results obtained in~\cite{zb}.
\begin{enumerate}
\item[(1)] Limit orders are placed, cancelled, and executed in a manner where jitters tend to be rapidly smoothed out, which is why we have a Laplacian;
\item[(2)] The change of the mid-price is driven by the intensity of interaction between ask and bid orders around the mid-price;
\item[(3)] The randomness comes from the constant creation, cancellation, and execution of limit orders; its intensity varies at different (limit) prices, and tends to vanish as the price goes far beyond the mid.
\end{enumerate}

In Section~\ref{S:MLO:ParamEst} we study the methods to estimate the model parameters based on a given limit order dataset and derive the statistics such as a Maximum-Likelihood Estimator (MLE) and an estimator that minimizes the Mean-Square Errors (MSE), both under AIC (see~\cite{aicm}) since the number of parameters (or dimension) is to be estimated itself. Under certain simplified (or degenerated) circumstances an explicit expression of an MLE estimator is also derived.

In Section~\ref{S:MLO:Optim} we study the investment optimization problem derived from the model, formulate and analyze the static and dynamic properties of the utility function of an investor who takes the model and its parameters as exogenous, which finally enable us to find the optimal limit-price-to-buy (or equivalently, amount-to-buy) of the asset and the corresponding amount of consumption, given a fixed amount of wealth and at a given time. Both static and dynamic analysis are studied to obtain the optimality of the investment via the limit order model, and two theorems are given respectively for those two types of analysis as the criteria to test for optimality from the model dynamics using It\^{o}'s Formula, which also have intuitive interpretations.

\section{Modeling}\label{S:MLO:Modeling}

\begin{figure}
\centering
\includegraphics[width=120mm]{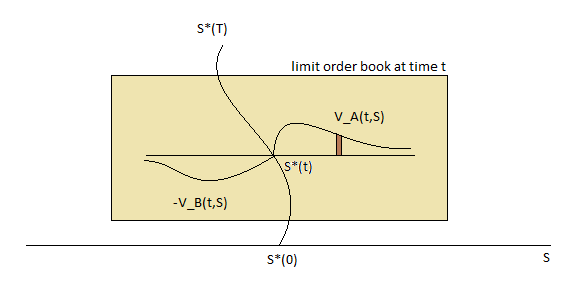}
\caption{The evolution model of the limit order book}
\label{Fig:mlo1}
\end{figure}

We model the evolution of the limit orders of a particular asset in the market. Suppose we work within the time interval $[0,T]$, and $S$ denotes the natural log of the price so $S\in\mathbb R$. At a particular time point $t\in[0,T]$, suppose the volume of the ask (resp. bid) limit orders on the limit order book from $S$ to $S+dS$ is $V_A(t,S)dS$ (resp. $V_B(t,S)dS$), and suppose the natural log of the mid-price is $S^*(t)$. Fix a probability space $(\Omega,\mathcal{F},\mathbb{P})$, We then have the following fairly natural but important model assumptions:
\begin{enumerate}
\item[(1)] all matched orders are executed immediately since all major trading centers today have been computerized, then we have for all $(S,t)$ such that $S>S^*(t)$, $V_B(S,t)=0$; similarly for all $(S,t)$ such that $S<S^*(t)$, $V_A(S,t)=0$;
\item[(2)] since jitters have a tendency to be rapidly smoothed out, $\partial V_A/\partial t$ contains a component $\alpha_A\Delta V_A$ where $\alpha_A$ is a positive constant and $\Delta:=\partial^2/\partial S^2$ is the Laplacian; the same is true for $V_B(S,t)$ with $\alpha_B$;
\item[(3)] the change of mid-price is driven by the ``strength'' of the ask and bid orders placed around the mid-price, which implies we have a Stefan-type condition
    \begin{equation}\label{Eq:mloBC}
     \rho\frac{dS^*}{dt}(t)=\left[\frac{\partial V_A}{\partial S}(t,S^*(t)+)+\frac{\partial V_B}{\partial S}(t,S^*(t)-)\right]
     \end{equation}
    where $\rho$ is a constant;
\item[(4)] for $S\ge S^*(t)$ the ask orders are placed in a stochastic manner, hence another component of $\partial V_A/\partial t$ is $\sigma_A(|S-S^*(t)|)\frac{\partial^2 W}{\partial t \partial S}$ where $\sigma_A$ is a regular scaling function (defined in Chapter~\ref{Ch:sspw}) and $W:\Omega\times[0,T]\times \mathbb R\to\mathbb R$ is a standard 2-dimensional Brownian sheet; the same is true for $V_B$ with $S\le S^*(t)$ and $\sigma_B$.
\end{enumerate}

The model is illustrated in Figure~\ref{Fig:mlo1}.

In sum, for a set of given model parameters $u_{0A},u_{0B},\alpha_A,\alpha_B,\sigma_A,\sigma_B,\rho$, the model has
\begin{equation}\label{Eq:mlo}
\begin{split}
 \frac{\partial V_A}{ \partial t} &= \alpha_A\frac{\partial^2 V_A}{\partial S^2} +\sigma_A(|S-S^*(t)|)\frac{\partial^2 W}{\partial t\partial S},\forall S>S^*(t),\\
 V_A(t,S)&=0,\forall S\le S^*(t),\\
 V_A(0,S)&=u_{0A}(S),\\
 \frac{\partial V_B}{ \partial t} &= \alpha_B\frac{\partial^2 V_B}{\partial S^2} +\sigma_B(|S-S^*(t)|)\frac{\partial^2 W}{\partial t\partial S},\forall S<S^*(t),\\
 V_B(t,S)&=0,\forall S\ge S^*(t),\\
 V_B(0,S)&=u_{0B}(S),\\
 \rho\frac{dS^*}{dt}(t)&=\left[\frac{\partial V_A}{\partial S}(t,S^*(t)+)+\frac{\partial V_B}{\partial S}(t,S^*(t)-)\right].
\end{split}\end{equation}

\begin{theorem}
The solution $V_A,V_B,S^*$ to model~\eqref{Eq:mlo} exists and is unique for $0\le t<\tau$ where $\tau$ is a well-defined stopping time.
\end{theorem}
\begin{proof} We roughly follow the same procedure to show the existence and uniqueness as in~\cite{zb}. First we make the following transformation: $\tilde V_A(t,S):=V_A(t,S^*(t)+S),\tilde u_{0A}(S):=u_{0A}(S^*(0)+S),\tilde V_B(t,S):=V_B(t,S^*(t)-S),\tilde u_{0B}(S):=u_{0B}(S^*(0)-S)$. Then the same argument as in Theorem 3.2,~\cite{zb} shows that there exists unique $\beta_A,\beta_B$ that is the limit of the iteration
\[\begin{split}
\rho\dot\beta_A^n(t)&=\int_0^\infty \frac{\partial \tilde p}{\partial x}(t,0,y)\tilde u_{0A}(y)dy+\int_0^t\int_0^\infty \frac{\partial \tilde p}{\partial x}(t-s,0,y)\sigma_A(y)W_{\beta^n}(dyds),\\
\rho\dot\beta_B^n(t)&=\int_0^\infty \frac{\partial \tilde p}{\partial x}(t,0,y)\tilde u_{0B}(y)dy+\int_0^t\int_0^\infty \frac{\partial \tilde p}{\partial x}(t-s,0,y)\sigma_B(y)W_{\beta^n}(dyds),\\
\beta^n(t)&=\beta_A^n(t)-\beta_B^n(t).
\end{split}\]
Define $S^*(t):=\lim_n\beta^n(t)$ which is also unique. Then using the same argument as in Theorem 3.2,~\cite{zb}, we have that $\tilde V_A$ and $\tilde V_B$ exist and are unique. Also, using Kolmogorov's Continuity Theorem as in Lemma 3.1.1,~\cite{zb}, we have that almost surely,
\[\begin{split}
\rho\dot\beta_A(t)&=\frac{\partial V_A}{\partial S}(t,S^*(t)+),\\
\rho\dot\beta_B(t)&=-\frac{\partial V_B}{\partial S}(t,S^*(t)-),
\end{split}\]
where $0\le t<\tau:=\lim_{L\to\infty} \tau^L$, and $\tau^L:=\inf\{t\ge 0:|\dot\beta_A(t)|\ge L\textrm{ or }|\dot\beta_B(t)|\ge L\}$.
 Therefore the Stefan boundary condition ~\eqref{Eq:mloBC} holds. $\Box$

\end{proof}

\section{Parameter Estimation}\label{S:MLO:ParamEst}

In this section a statistical method based on AIC is developed to estimate the parameters of model~\eqref{Eq:mlo} given a real dataset of the limit order book and the mid-price of a particular asset within a certain period of time.

Suppose the dataset consists of 3 matrices, two $T\times N$ matrices $D_A$, $D_B$ for the volumes of the ask and bid orders, and a $T\times 1$ matrix $P$ for the mid-price. The sampling steps for time and price are $\Delta_T$ and $\Delta_N$, that is, the dataset is from time $0$ to $\Delta_T T$, and the $t$-th row of $D_A$ stores the volumes of the ask orders from limit price $P[t]$ to $P[t]+\Delta_N N$, while the $t$-th row of $D_B$ stores the volumes of the bid orders from limit price $P[t]-\Delta_N N$ to $P[t]$, where $[\cdot]$ denotes the vector subscription.

Our goal in this section is to develop an algorithm to numerically compute or even find an explicit expression of the statistics served as the appropriate estimators of the model parameters under maximum-likelihood, minimum mean-square errors, and AIC. The method is decomposed into two parts, first parameters $\alpha_A,\alpha_B,\sigma_A,\sigma_B$ are estimated using maximum-likelihood approach combined with AIC based on the limit order book data $D_A$ and $D_B$, and then $u_{0A},u_{0B},\rho$ are estimated using least mean-square-errors combined with AIC based on the whole dataset $D_A,D_B,P$.

\subsection{AIC/MLE Estimation of $\alpha_A,\alpha_B,\sigma_A,\sigma_B$}

To estimate $\sigma_A$ and $\sigma_B$, we further assume that
\[ \sigma_i(x):=\frac{x^{1.6}}{1+xp_i(x)},i=A,B\]
where for $i=A,B$, $p_i$ is a polynomial with $\deg p_i\ge1$, so that $\sigma_i$ is a regular scaling function, and $\lim_{x\to\infty} \sigma_i(x)=0$, because there tends to be very little randomness as the limit price goes far beyond the mid-price. Then the goal is to estimate $\alpha_A,\alpha_B$ and the degrees of $p_A,p_B$ and their coefficients. Since the number of parameters $4+\deg p_A+\deg p_B$ is also to be estimated, an AIC-based approach is used.

Suppose $i=A,B$, and $p_i(x)=\sum_{j=0}^{d_i} p_{ij}x^j$ with $d_i\ge 1$. To estimate $\alpha_i,p_{i0},\ldots,p_{id_i}$, we use a finite-difference Euler approximation scheme similar in Section 4,~\cite{zb}. For each integer $n$, define the difference operator $\nabla:\mathbb R^n\to\mathbb R^{n-1}$ by
\[ \nabla\left((a_1,\ldots,a_n)^T\right)=\left(a_2-a_1,\ldots,a_n-a_{n-1}\right)^T.\]
Denote $\left(\nabla a\right)[i]$ by $\nabla a_i$. For a matrix with row label $R$ and column label $C$ denote $\nabla_R$ as the difference operator on column vectors and $\nabla_C$ as that on row vectors. Then for each $i=A,B,t=1,\ldots,T-1, S=1,\ldots,N-2$,
\[ \frac{\nabla_t D_i[t,S]}{\Delta_T\Delta_N} =\alpha_i \frac{\nabla_S^2D_i[t,S]}{\Delta_N^3}+\sigma_i(S\Delta_N)\frac{\xi_{t,S}}{\Delta_T\Delta_N}\]
where
\[ \xi_{t,S}:=W\left(\left[\Delta_Tt,\Delta_T(t+1)\right]\times\left[S^*(\Delta_Tt)+\Delta_NS,S^*(\Delta_Tt)+\Delta_N(S+1)\right]\right)\] are independent identically distributed Gaussian random variables with mean $0$ and variance $\Delta_T\Delta_N$. Fix $d_i=\deg p_i$, then the log likelihood function $\ell_i(\alpha_i,p_{i0},\ldots,p_{id_i})$ satisfies
\begin{equation}\label{Eq:mloEst1}
-2\ell_i(\alpha_i,p_{i0},\ldots,p_{id_i})=\sum_{t=1}^{T-1}\sum_{S=1}^{N-2}\left\{\left[\nabla_t D_i[t,S]-\alpha_i \left(\frac{\Delta_T}{\Delta_N^2}\right)\nabla_S^2D_i[t,S]\right]\frac{1+\sum_{j=0}^{d_i} p_{ij}S^{j+1}\Delta_N^{j+1}}{S^{1.6}\Delta_N^{1.6}}\right\}^2.
\end{equation}
When $d_i$ is fixed, the first order condition is a cubic formula of the variables, which shall be solved using numerical methods. If the model can be degenerated so that $\alpha_i=\alpha_{0i}$ is a known constant, then we can indeed solve for optimal $\hat p_{i0},\ldots,\hat p_{id_i}$:
\[ \textbf{p}_i=-\textbf{A}_i^{-1}\textbf{b}_i,\]
where $\textbf{p}_i=(\hat p_{i0},\ldots,\hat p_{id_i})^T$, $\textbf{A}_i$ is a $(d_i+1)$-square matrix with its $(m,n)$ element being
\[ \sum_{t=1}^{T-1}\sum_{S=1}^{N-2}\left\{\left[\nabla_t D_i[t,S]-\alpha_{0i} \left(\frac{\Delta_T}{\Delta_N^2}\right)\nabla_S^2D_i[t,S]\right]S^{m+n-3.2}\Delta_N^{m+n}\right\}^2,
\]
and $\textbf{b}_i$ is a $(d_i+1)\times 1$ matrix with its $m$-th element being
\[
\sum_{t=1}^{T-1}\sum_{S=1}^{N-2}\left\{\left[\nabla_t D_i[t,S]-\alpha_{0i} \left(\frac{\Delta_T}{\Delta_N^2}\right)\nabla_S^2D_i[t,S]\right]S^{2m-3.2}\Delta_N^{2m}\right\}^2.
\]

Therefore, under AIC (see~\cite{aicm}), the set of optimal estimators $\hat \alpha_i,\hat d_i, \hat p_{i0},\ldots,\hat p_{id_i}$ minimizes
\begin{equation}\label{Eq:mloAIC1}
 \textrm{AIC}_i:=2d_i-2\ell_i(\alpha_i,p_{i0},\ldots,p_{id_i}).
 \end{equation}
The algorithm to minimize $\textrm{AIC}_i$:
\begin{enumerate}
\item[(1)] $m:=-\infty, \hat d_i:=0$;
\item[(2)] For $d_i\in\{1,\ldots,10\}$:
\begin{enumerate}
\item[(a)] Numerically solve the first order condition of~\eqref{Eq:mloEst1} for optimal $\hat \alpha_i, \hat p_{i0},\ldots,\hat p_{id_i}$;
\item[(b)] Calculate $a:=$AIC($d_i$) based on~\eqref{Eq:mloAIC1};
\item[(c)] If $m>a$ then $m:=a, \hat d_i:=d_i$;
\end{enumerate}
\item[(3)] The final $\hat d_i$ is optimal within $\{1,\ldots,10\}$ with AIC $m$.
\end{enumerate}

\subsection{AIC/Least-MSE Estimation of $u_{0A},u_{0B},\rho$}

We continue to assume $i=A,B$. Similarly we need further assumptions about the initial condition $u_{0i}$. Assume
\[ u_{0i}(x)=xq_i(x)\exp(-\gamma_i x)\]
where $q_i$ is a polynomial and $\gamma_i>0$ is a parameter, so that $u_{0i}$ has exponential decay at infinity. Then $c_i:=\deg q_i$ and its coefficients $q_{i0},\ldots,q_{ic_i}$ are also a model parameter that needs to be determined under AIC. Now, if we take the expectation conditional on $S^*$ on both sides of the evolution equation of $V_i$, we get
\[ \mathbb E[V_i(t,S)|S^*]=\int_0^\infty p_-(t,S,y)u_{0i}(y)dy+\int_0^t S^*(s)\int_0^\infty q(t-s,S,y)\mathbb E[V_i(s,y)|S^*]dyds.\]
This implies that $\mathbb E[V_i|S^*]$ satisfies a deterministic Stefan-type PDE, which means that although we are generally unable to find an explicit expression of the solution, we can find a sufficiently accurate numerical solution. For a given dataset $D_A, D_B, P$, the goal in this part is thus to develop an algorithm to minimize the mean-square errors against $\mathbb E[V_i|S^*]$ and the mid-price function (moving boundary) determined by it under AIC.

Given a set of parameters $c_A, c_B, q_{A0},\ldots,q_{Ac_A},q_{B0},\ldots,q_{Bc_B},\gamma_A,\gamma_B,\rho$, denote the solution to the deterministic Stefan-type PDE (i.e., when $W\equiv 0$) by $\bar V_A$ and $\bar V_B$, which can be obtained numerically. Then the variances (errors) come from two parts: those from the evolution of the limit order book, and those from the evolution of the mid-price. Let $\theta_0$ be a predefined constant which is the weight of importance of how well mid-price fits against how well the limit order book fits. Namely, we shall minimize the weighted MSE
\[ \textrm{MSE}:=\textrm{MSE}_1+\theta_0\textrm{MSE}_2\]
where
\[ \textrm{MSE}_1:=\frac{1}{2TN}\sum_{i\in\{A,B\}}\sum_{t=1}^T \sum_{S=1}^N \left[D_i(t,S)-\bar V_i(\Delta_T t,\Delta_N S)\right]^2,\]
and
\[ \textrm{MSE}_2:=\frac{1}{T}\sum_{t=1}^T \left[\rho P(t)-\frac{\bar V_A(\Delta_T t,S)+\bar V_B(\Delta_T t,S)}{\Delta_T}\right]^2.\]
Finally, similar to the AIC approach in the previous part, the goal is to minimize
\[ c_A+c_B+\textrm{MSE},\]
which can be done by traversing all possible values of the model parameters, numerically finding $\bar V_A$ and $\bar V_B$, computing the corresponding MSE's, and finally finding the recorded optimal parameters. Note that unlike in previous part, because of the nonlinearity of $\bar V_A$ and $\bar V_B$, it is difficult to find a direct way to (numerically) solve for optimal coefficients for $q_i$, as opposed to $p_i$.

\section{Investment Optimization}\label{S:MLO:Optim}

Using the model combined with a set of optimal parameters adjusted for a given dataset, an investor can optimize his or her allocation of the amount of investment in the asset against consumption within a given amount of wealth.  To analyze the investment behavior under our model, we assume the investor can only buy a single asset by making limit order transactions, that is, fulfilling ask orders placed by other sellers. This scenario is typical when an investor is avert to the volatility of price changes in market order transactions, or there lacks sufficient liquidity of the asset but the investor still has a strong motivation to make purchases.

Let $U:\mathbb R_+^3\to\mathbb R$ be the utility function of the investor, which is monotone increasing and concave down on the second and third parameters. Specifically, $U(t,L_t,C_t)$ denotes the amount of happiness the investor obtains at time $t$ if having the amount of asset $L_t$ and consumption $C_t$. Since the investor would choose to pay the lowest possible price to buy an amount of asset, the amount of asset $L_t$ and the total cost to buy such amount is completely determined by the highest limit price $S^*(t)+B$ the investor is willing to pay, which is illustrated in Figure~\ref{Fig:mlo2}.

\begin{figure}
\centering
\includegraphics[width=120mm]{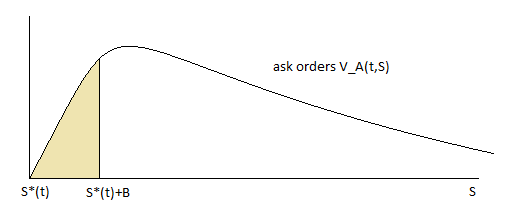}
\caption{Investment Optimization via Limit Orders}
\label{Fig:mlo2}
\end{figure}

Taking the model (and its parameters) as exogenous, and assuming the time $t$ wealth is $W_t$, we then have the following optimization problem at time $t$ (where the subscription $t$ in the expectation denotes it is conditional to all information up to time $t$):
\[ \max_{B\ge 0} \mathbb E_t\left[U\left(t,L_t(B),C_t\right)\right] \]
where
\[ L_t(B)=\int_{S^*(t)}^{S^*(t)+B} V_A(t,S)dS \]
subject to the budget constraint
\[ W_t = C_t + \int_{S^*(t)}^{S^*(t)+B}  S V_A(t,S) dS. \]

When $t$ is fixed, the optimization of $U$ with respect to the choice of $B$ is done through a static analysis; if we consider the full model, that is, both $t$ and $B$ vary, we come up with a dynamic analysis.

\subsection{Static Analysis}

Let time $t$ be fixed. Denote the partial derivatives of $U$ by $U_t, U_L,U_C$. Then we compute
\[\begin{split} \frac{\partial U}{\partial B}(t,L_t(B),C_t) &= U_L(t,L_t(B),C_t)V_A(t,S^*(t)+B)\\
&\quad\quad -U_C(t,L_t(B),C_t)(S^*(t)+B)V_A(t,S^*(t)+B).
\end{split}\]
Therefore the optimal $B=B^*$ satisfies the first order condition
\[ S^*(t)+B^*=\frac{U_L(t,L_t(B^*),C_t)}{U_C(t,L_t(B^*),C_t)}.\]
This means
\begin{theorem}\label{Thm:mlo1} In static optimization where the time $t$ is fixed, the optimal highest limit price to buy the asset is equal to the ratio of the marginal utility of amount of asset to that of consumption.
\end{theorem}
\begin{proof} As above. $\Box$
\end{proof}
Note that $S^*(t)+B^*$ can be seen as the highest amount of money the investor is willing to pay to substitute $1$ unit of asset with the same amount of consumption.

\subsection{Dynamic Analysis}

Now consider that $t$ also varies, and we are interested in $\mathbb E_t[dU]$ from $t$ to $t+dt$. Indeed, if we have $\mathbb E_t[dU(t,L_t(B^*),C_t)]>0$, then even if $B=B^*$ reaches its static optimality, the investor would still wait for an amount of time to maximize the utility.

First we consider the evolution of $L_t(B)$. Fixing $B$ and plugging in $dV_A$, we have formally
\[\begin{split}
dL_t(B)&=\alpha_A\left[\frac{\partial V_A}{\partial S}(t,S^*(t)+B)-\frac{\partial V_A}{\partial S}(t,S^*(t))\right]dt\\
&\quad\quad+\int_{S^*(t)}^{S^*(t)+B} \sigma(S-S^*(t))W(dSdt)\\
&\quad\quad+\left[\int_0^t \int_{S^*(t)}^{S^*(t)+B} \sigma(S-S^*(t))W(dSdt)\right]\dot S^*(t)dt.
\end{split}\]
Similarly,
\[\begin{split}
dC_t&=\alpha_A\left[V_A(t,S^*(t)+B)-(S^*(t)+B)\frac{\partial V_A}{\partial S}(t,S^*(t)+B)+S^*(t)\frac{\partial V_A}{\partial S}(t,S^*(t))\right]dt\\
&\quad\quad+\int_{S^*(t)}^{S^*(t)+B} S\sigma(S-S^*(t))W(dSdt)\\
&\quad\quad+\left[\int_0^t \int_{S^*(t)}^{S^*(t)+B} S\sigma(S-S^*(t))W(dSdt)\right]\dot S^*(t)dt.
\end{split}\]

Then by It\^{o}'s formula, and noting that $S^*(t)$ is part of the information at $t$, we get
\[\begin{split}
 \mathbb E_t[dU]&=U_t dt+U_L\alpha_A\left[\frac{\partial V_A}{\partial S}(t,S^*(t)+B)-\frac{\partial V_A}{\partial S}(t,S^*(t))\right]dt+\frac12 U_{LL}\int_0^B\sigma^2(S)dSdt\\
 &\quad\quad+U_C\alpha_A\left[V_A(t,S^*(t)+B)-(S^*(t)+B)\frac{\partial V_A}{\partial S}(t,S^*(t)+B)+S^*(t)\frac{\partial V_A}{\partial S}(t,S^*(t))\right]dt\\
 &\quad\quad+\frac12U_{CC}\int_0^B (S^*(t)+S)^2\sigma^2(S)dSdt.
 \end{split}\]
Since for a fixed $t$ we are only interested at $B^*$ in Theorem~\ref{Thm:mlo1}, and $S^*+B^*=U_L/U_C$, we have
\begin{equation}\label{Eq:mloDyn}
\begin{split}
 \mathbb E_t\left[\frac{dU}{dt}(t,L_t(B^*),C_t)\right]&=U_t+U_C\alpha_A\left[V_A(t,S^*(t)+B^*)-B^*\frac{\partial V_A}{\partial S}(t,S^*(t))\right]\\
 &\quad\quad+\frac12 U_{LL}\int_0^{B^*}\sigma^2(S)dS+\frac12U_{CC}\int_0^{B^*}(S^*(t)+S)^2\sigma^2(S)dS.
 \end{split}\end{equation}

If we denote the absolute risk aversions of the investor with respect to the amount of asset and the consumption by
\[\begin{split}
 r_L&:=-\frac{U_{LL}}{U_L},\\
 r_C&:=-\frac{U_{CC}}{U_C},
\end{split}\]
then~\eqref{Eq:mloDyn} is equivalent to
\begin{equation}\label{Eq:mloDynR}
\begin{split}
 \mathbb E_t\left[\frac{dU}{dt}(t,L_t(B^*),C_t)\right]&=U_t+U_C\left\{\alpha_A\left[V_A(t,S^*(t)+B^*)-B^*\frac{\partial V_A}{\partial S}(t,S^*(t))\right]\right.\\
 &\quad\quad\left.-\frac12 r_L(S^*+B^*)\int_0^{B^*}\sigma^2(S)dS-\frac12r_C\int_0^{B^*}(S^*+S)^2\sigma^2(S)dS\right\}.
 \end{split}\end{equation}

The investor would then use the dataset $D_A$ to compute $\mathbb E_t\left[\frac{dU}{dt}\right]$ and see the expected change of $U$ from $t$ to $t+dt$ at $B^*$. Specifically, we have the following observations from Equation~\eqref{Eq:mloDyn} or~\eqref{Eq:mloDynR}:
\begin{enumerate}
\item[(1)] the larger the quantity
\[ V_A(t,S^*(t)+B^*)-B^*\frac{\partial V_A}{\partial S}(t,S^*(t))\]
is, the more possible it is for $U$ to increase;
\item[(2)] the greater absolute risk aversions $r_L$ and $r_C$ the investor has, the more possible it is for $U$ to decrease;
\item[(3)] the larger $B^*$ is, the more risk $U$ is exposed to for a decrease, since both $\int_0^{B^*}\sigma^2(S)dS$ and $\int_0^{B^*} (S^*+S)^2\sigma^2(S)dS$ are increasing functions of $B^*$.
\end{enumerate}

Also, if we make two natural assumptions about $U$,
\begin{enumerate}
\item[(a)] the utility function is discounted in time, that is, $U_t<0$,
\item[(b)] the investor is risk avert or risk neutral, that is, $U_{LL}\le 0,U_{CC}\le 0$,
\end{enumerate}
then we have
\begin{theorem}\label{Thm:mloDyn1} Suppose $B^*$ is the statically optimal choice at time $t$ as in Theorem~\ref{Thm:mlo1}. If
\begin{equation}\label{Eq:mloDyn1}
 V_A(t,S^*(t)+B^*)\le B^*\frac{\partial V_A}{\partial S}(t,S^*(t)),
 \end{equation}
then
\[ \mathbb E_t \left[\frac{dU}{dt}(t,L_t(B^*),C_t)\right]\le0.\]
\end{theorem}
\begin{proof} As above. $\Box$
\end{proof}

\begin{figure}
\centering
\includegraphics[width=100mm]{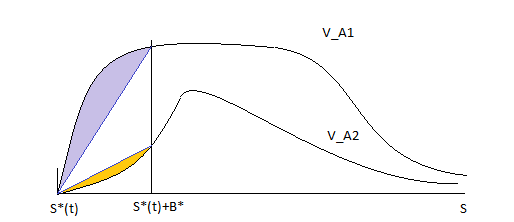}
\caption{An Illustration of Theorem~\ref{Thm:mloDyn1}}
\label{Fig:mlo3}
\end{figure}

Note that~\eqref{Eq:mloDyn1} in Theorem~\ref{Thm:mloDyn1} has an intuitive illustration in Figure~\ref{Fig:mlo3}. The quantity $V_A(t,S^*(t)+B^*)/B^*$ is the slope of the colored lines between the mid-price point $(S^*(t),0)$ and the optimal point $(S^*+B^*,V_A(t,S^*+B^*))$.
 \begin{enumerate}
 \item[*] If it is less than the boundary derivative $\partial V_A(t,S^*(t))/\partial S$ as shown in curve $V_{A1}$, then the utility at next instant $t+dt$ is expected to drop, as the volume of ask orders at lower limit prices between $S^*$ and $S^*+B^*$ tends to fall, so it is best to make the purchase at $t$ rather than wait.
\item[*] On the other hand, as shown in curve $V_{A2}$, the volume at lower prices tends to rise, and it might be wise to wait for a future time to make the purchase. In this case the investor needs to evaluate other factors such as time discount $U_t$ and risk premium terms as well.
\end{enumerate}

\section{Summary}

In this chapter we summarize the main results obtained in the previous sections. Throughout this chapter we fix a probability space $(\Omega,\mathcal F,\mathbb P)$, and suppose $W:\Omega\times\mathbb R\times\mathbb R\to\mathbb R$ is a standard 2-dimensional Brownian sheet.

\subsection{Modeling}

For a set of given model parameters $u_{0A},u_{0B},\alpha_A,\alpha_B,\sigma_A,\sigma_B,\rho$, the model
\begin{equation}
\begin{split}
 \frac{\partial V_A}{ \partial t} &= \alpha_A\frac{\partial^2 V_A}{\partial S^2} +\sigma_A(|S-S^*(t)|)\frac{\partial^2 W}{\partial t\partial S},\forall S>S^*(t),\\
 V_A(t,S)&=0,\forall S\le S^*(t),\\
 V_A(0,S)&=u_{0A}(S),\\
 \frac{\partial V_B}{ \partial t} &= \alpha_B\frac{\partial^2 V_B}{\partial S^2} +\sigma_B(|S-S^*(t)|)\frac{\partial^2 W}{\partial t\partial S},\forall S<S^*(t),\\
 V_B(t,S)&=0,\forall S\ge S^*(t),\\
 V_B(0,S)&=u_{0B}(S),\\
 \rho\frac{dS^*}{dt}(t)&=\left[\frac{\partial V_A}{\partial S}(t,S^*(t)+)+\frac{\partial V_B}{\partial S}(t,S^*(t)-)\right]
\end{split}\end{equation}
exists and is unique.

\subsection{Parameter Estimation}

Suppose $P,D_A,D_B$ is a given limit order dataset. For $i=A,B$, suppose
\[ \sigma_i(x):=\frac{x^{1.6}}{1+xp_i(x)},i=A,B\]
where $p_i=\sum_{j=0}^{d_i} p_{ij}x^j$ is a polynomial with $\deg p_i\ge1$. Then the first step is to minimize for $i=A,B$,
\begin{equation}\begin{split}
\textrm{AIC}_i&:=2d_i-2\ell_i(\alpha_i,p_{i0},\ldots,p_{id_i})\\
&=\sum_{t=1}^{T-1}\sum_{S=1}^{N-2}\left\{\left[\nabla_t D_i[t,S]-\alpha_i \left(\frac{\Delta_T}{\Delta_N^2}\right)\nabla_S^2D_i[t,S]\right]\frac{1+\sum_{j=0}^{d_i} p_{ij}S^{j+1}\Delta_N^{j+1}}{S^{1.6}\Delta_N^{1.6}}\right\}^2.
\end{split}\end{equation}
Suppose also
\[ u_{0i}(x)=xq_i(x)\exp(-\gamma_i x)\]
where $q_i=\sum_{j=0}^{c_i} q_{ij}x^j$ is a polynomial and $\gamma_i>0$ is a parameter. Then the second step is for a predefined weight $\theta_0$ to minimize
\[ c_A+c_B+\textrm{MSE}_1+\theta_0\textrm{MSE}_2\]
where
\[\begin{split}
\textrm{MSE}_1&:=\frac{1}{2TN}\sum_{i\in\{A,B\}}\sum_{t=1}^T \sum_{S=1}^N \left[D_i(t,S)-\bar V_i(\Delta_T t,\Delta_N S)\right]^2,\\
\textrm{MSE}_2&:=\frac{1}{T}\sum_{t=1}^T \left[\rho P(t)-\frac{\bar V_A(\Delta_T t,S)+\bar V_B(\Delta_T t,S)}{\Delta_T}\right]^2,
\end{split}\]
and $\bar{V}_i(t,x)$ is the numerical solution by setting $W\equiv0$.

\subsection{Optimization}

Suppose $B^*(t)$ is the statically optimal price at time $t$ with respect to the utility function $U$ and wealth $W_t$. Then
\[ S^*(t)+B^*(t)=\frac{U_L(t,L_t(B^*(t)),C_t)}{U_C(t,L_t(B^*(t)),C_t)}.\]
In other words, the optimal highest limit price to buy the asset is equal to the ratio of the marginal utility of amount of asset to that of consumption.

Also, let $r_L,r_C$ be the absolute risk aversions with respect to $L,C$, then we have
\begin{equation}
\begin{split}
 \mathbb E_t\left[\frac{dU}{dt}(t,L_t(B^*),C_t)\right]&=U_t+U_C\left\{\alpha_A\left[V_A(t,S^*(t)+B^*)-B^*\frac{\partial V_A}{\partial S}(t,S^*(t))\right]\right.\\
 &\quad\quad\left.-\frac12 r_L(S^*+B^*)\int_0^{B^*}\sigma^2(S)dS-\frac12r_C\int_0^{B^*}(S^*+S)^2\sigma^2(S)dS\right\}.
 \end{split}\end{equation}
In particular, if
\begin{equation}
 \frac{V_A(t,S^*(t)+B^*)}{B^*}\le \frac{\partial V_A}{\partial S}(t,S^*(t)),
 \end{equation}
then
\[ \mathbb E_t \left[\frac{dU}{dt}(t,L_t(B^*),C_t)\right]\le0.\]

\bibliographystyle{plain}
\bibliography{thesisbib}

\end{document}